\documentclass[10pt]{amsart}

\usepackage{hyperref}

\usepackage{amssymb}
\usepackage{amsmath}

\let\newpf\proof \let\proof\relax

\def\bm{\begin{matrix}}
\def\em{\end{matrix}}

\newcommand{\bt}{\begin{thm}}
\newcommand{\et}{\end{thm}}

\newcommand{\bl}{\begin{lemma}}
\newcommand{\el}{\end{lemma}}

\newcommand{\beq}{\begin{eqnarray}}
\newcommand{\eeq}{\end{eqnarray}}

\def\be{\begin{equation}}
\def\ee{\end{equation}}

\def\ba{{\begin{align}}}
\def\ea{{\end{align}}}

\def\0{{\mathbf 0}}

\newtheorem{thm}{Theorem}[section]

\newtheorem{lemma}[thm]{Lemma}

\newtheorem{prop}[thm]{Proposition}

\theoremstyle{remark}
\newtheorem{rem}{Remark}[section]

\numberwithin{equation}{section}

\def \bn {\hfill \\ \smallskip\noindent}

\theoremstyle{definition}

\newtheorem{definition}{Definition}[section]
\def\proof{\bn {\bf Proof.} }

\def\note#1
{\marginpar

{\tiny $\leftarrow$
\par
\hfuzz=20pt \hbadness=9000 \hyphenpenalty=-100 \exhyphenpenalty=-100
\pretolerance=-1 \tolerance=9999 \doublehyphendemerits=-100000
\finalhyphendemerits=-100000 \baselineskip=6pt
#1}\hfuzz=1pt}

\newcommand{\diam}{\operatorname{diam}}
\newcommand{\dist}{\operatorname{dist}}

\renewcommand{\mod}{\operatorname{mod}}

\newcommand{\N}{{\mathbb N}}
\newcommand{\Q}{{\mathbb Q}}
\newcommand{\R}{{\mathbb R}}
\newcommand{\T}{{\mathbb T}}

\newcommand{\Z}{{\mathbb Z}}

\def\B0{{\bold{0}}}

\catcode`\@=12

\def\Empty{}
\newcommand\oplabel[1]{
  \def\OpArg{#1} \ifx \OpArg\Empty {} \else
  	\label{#1}
  \fi}

\newcommand{\comm}[1]{}
\newcommand{\comment}[1]{}

\begin{document}

\title{Generic continuous spectrum for multi-dimensional quasiperiodic Schr\"odinger operators with rough potentials}
\title[]{Generic continuous spectrum for multi-dimensional quasiperiodic Schr\"odinger operators with rough potentials}

\author{Rui Han}

\author{Fan Yang}
\thanks{}

\begin{abstract}
We study the multi-dimensional operator $(H_x u)_n=\sum_{|m-n|=1}u_{m}+f(T^n(x))u_n$, where $T$ is the shift of the torus $\T^d$. 
When $d=2$, we show the spectrum of $H_x$ is almost surely purely continuous for a.e. $\alpha$ and generic continuous potentials.
When $d\geq 3$, the same result holds for frequencies under an explicit arithmetic criterion.
We also show that general multi-dimensional operators with measurable potentials do not have eigenvalue for generic $\alpha$.
\end{abstract}

\maketitle

\section{Introduction}
In this note, we are interested in quasiperiodic Schr\"odinger operators acting on $l^2(\Z^d)$ as follows
\begin{align}\label{defH}
(H_x u)_n=\sum_{|m-n|=1}u_{m}+f(T^nx)u_n,
\end{align}
where $f:\T^d\rightarrow\R$, $T^nx=(x_1+n_1\alpha_1,...,x_d+n_d\alpha_d)$, $|m|=\sum_{j=1}^d|m_j|$. We refer to $f$ as the potential, $\alpha\in (\T\backslash \Q)^d$ as the frequency and $x\in \T^d$ as the phase.
It is known that the spectral types of $H_x$ are almost surely independent of $x$. 
By the well known RAGE theorem, different spectral types lead to different long-time behaviour of the solutions to time-dependent Schr\"odinger equations.
Thus one is naturally interested in finding the way to identify the spectral types of a given operator. 
In this note we will show that for multi-dimensional operators, purely continuous spectrum is a generic\footnote{In this note, generic means dense $G_{\delta}$.} phenomenon.

In the one-dimensional case, $d=1$, a very useful criterion to exclude point spectrum, due to Gordon \cite{G0}, states that if the potential can be approximated in a reasonably fast sense by periodic ones, then the operator does not have point spectrum. 
Following this idea, Boshernitzan and Damanik showed that given any $\alpha$, for a generic continuous potential, $H_x$ has empty point spectrum for a.e. $x$ \cite{BD}. 
In the first part of this note, we generalize this result to multi-dimensional case. 
When $d=2$, we show this result holds for an explicit full measure set of $\alpha$, when $d\geq 3$, we show this result is true under an explicit arithmetic criterion for $\alpha$.
It is interesting whether this result holds for a.e. $\alpha$ in the $d \geq 3$ case.

In the second part of this note, we explore the generic phenomenon in the frequency space. 
According to Simon's Wonderland Theorem \cite{Wonder}, dense continuous spectrum for a large class of metric spaces of operators implies generic continuous spectrum.
Fixing a continuous potential, since convergence in frequency implies operator convergence in the strong resolvent sense.
Thus, by Wonderland Theorem, generic continuous spectrum follows from continuous spectrum for rational frequencies.
However discontinuous potentials do not fall into the criterion of Wonderland Theorem.
It was recently observed by Gordon \cite{G1} that when $d=1$\footnote{The author actually dealt with one-dimensional operator but with multi-dimensional frequency.}, one can prove continuous spectrum for generic frequency even for measurable potentials.
In this note we generalize this result to multi-dimensional operators.

A key ingredient that enables us to deal with multi-dimensional operators is a criterion recently discovered by Gordon and Nemirovski \cite{G2}.

Our results for generic continuous potentials are as follows
\begin{thm}\label{2d}
When $d=2$, if $(\alpha_1, \alpha_2)$ are not both of bounded type, then for generic continuous potentials $f$, $H_x$ has no point spectrum for a.e. $x\in \T^2$.
\end{thm}
The proof of Theorem \ref{2d} relies on the following result about general multi-dimensional operators. Let $\|x\|_{\T}=\dist{(x, \Z)}$.
\begin{thm}\label{genericae}
Suppose there exists an infinite sequence $\mathcal{Q} =\{ \tau^{(n)}=(\tau_1^{(n)}, \cdots, \tau_d^{(n)}) \}$ such that
\begin{align}\label{GenericA}
\lim_{n\rightarrow\infty}\frac{\tau_1^{(n)}\cdots \tau_d^{(n)}}{\tau_i^{(n)}}\|\tau_i^{(n)}\alpha_i\|_{\T}=0\ \ \mathrm{for}\ \ \mathrm{any}\ \ i=1,...,d.
\end{align}
Then for generic continuous potentials $f$, $H_x$ has no point spectrum for a.e. $x\in \T^d$.
\end{thm}
\begin{rem}\label{remlabel}
For $d=2$, (\ref{GenericA}) holds if and only if $(\alpha_1, \alpha_2)$ are not both of bounded type (see definition \ref{defbdd}), see Lemma \ref{iff} in section \ref{2dproof}. However for $d\geq 3$, (\ref{GenericA}) only holds for Lebesgue measure zero set of frequencies due to a simple argument by Borel-Cantelli lemma, see section \ref{rem}. 
\end{rem}

Our result for measurable potentials is as follows.
\begin{thm}\label{measuregeneric}
Let the potential $f$ be a measurable function. For generic $\alpha\in\T^d$, $H_x$ has empty point spectrum for a.e. $x$.
\end{thm}
\begin{rem}
This theorem could be easily generalized to the case of arbitrary (not necessarily equal to $d$) number of frequencies.
\end{rem}

We organize this note as follows: section 2 serves as a preparation for our proofs of Theorems \ref{genericae}, \ref{measuregeneric} in sections 3 and 5. The proof of Theorem \ref{2d} will be discussed in section 4. In section 6 we show the simple argument by Borel-Cantelli lemma that we mentioned in Remark \ref{remlabel}.

\section{Preliminaries}
For a Borel set $U\subset \R$, we let $|U|$ be its Lebesgue measure. For $x\in \R^d$, let $\|x\|_{\T^d}=\dist{(x, \Z^d)}$. For a measurable function $f$, let $\|f\|_{\infty}=\{\inf M\geq 0: |f(x)|\leq M\ \ \mathrm{for} \ a.e.x\}$ be the $L_{\infty}$ norm.
\subsection{Some facts from measure theory}\
Let $f$ be a measurable function on $\T^d$. It is known that $f(\cdot +y)$ converges to $f(\cdot)$ in measure as $\T^d \ni y \rightarrow 0$. 
We set
 \begin{equation}\label{defsetF}
F(y,\epsilon)=\{x\in \T^d: |f(x+y)-f(x)|\geq \epsilon\}.
\end{equation} 
Then we have the following fact
\begin{prop}
For any $\epsilon >0$ and any $\eta>0$, there is  $\kappa(\epsilon, \eta)>0$ such that
if \begin{equation}\label{defkappa}
\|y\|_{T^d} <\kappa(\epsilon, \eta)
\end{equation}
then we have $|F(y, \epsilon)| < \eta$.
\end{prop}
We will also set 
\begin{equation}\label{defEM}
E(M)=\{x\in T^d:|f(x)|>M\}
\end{equation}
Clearly, $|E(M)| \rightarrow 0$ as $M \rightarrow \infty$.

\subsection{Continued fraction approximants}
Let $\{\frac{p_n}{q_n}\}$ be continued fraction approximants of $\alpha$. The following properties of continued fractions will be used later. First,
\begin{align}\label{qnan}
\left\lbrace
\begin{matrix}
p_{n+1}=a_np_n+p_{n-1}\\
q_{n+1}=a_nq_n+q_{n-1}
\end{matrix}
\right.
\end{align}
Secondly, for any $1\leq k\leq q_n-1$, we have
\begin{align}\label{qnbest}
\|q_n\alpha\|_{\T}\leq \|k\alpha\|_{\T}.
\end{align}
Thirdly, we have
\begin{align}\label{qnqn+1}
\frac{1}{q_{n+1}}\leq \|q_n\alpha\|_{\T}\leq \frac{2}{q_{n+1}}.
\end{align}
\begin{definition}\label{defbdd}
$\alpha$ is said to be of {\it bounded type} if there exists $C>0$ such that $a_n\leq C$ for any $n\in \N$.
\end{definition}
\begin{rem}
It is well known that bounded type $\alpha$ form a Lebesgue measure zero set.
\end{rem}
If $\alpha$ is of bounded type, by (\ref{qnan}), (\ref{qnbest}) and (\ref{qnqn+1}), we have for some $C>0$, 
\begin{align}\label{boundedtype}
\|k\alpha\|_{\T}\geq \frac{1}{Ck}\ \ \mathrm{for}\ \ \mathrm{any}\ \ k\geq 1.
\end{align}

\subsection{A key ingredient from \cite{G2}}
We have combined Theorems 3.1 and 5.1 from \cite{G2} into the following form, which is more convenient for us to use in this note.
\begin{thm}\label{gordon}
Let $V$  be a complex-valued function on $\Z^d$. 
Suppose there exists $\gamma>\delta>0$ and an infinite set $\mathcal{P}\subset \N^d$ satisfying 
$$
lim_{\mathcal{P} \ni \tau \rightarrow \infty} \tau_i =\infty,\   i=1, \cdots, d,
$$
such that there is a ($\tau_1, \cdots, \tau_d$)-periodic function $V_{\tau}(\cdot)$ satisfying the property that for some $\lambda_0>0$,
$$
\rho_{\tau}<(2d-1+M_{\tau}+\lambda_0)^{-(2d+\gamma)\tau_1\cdots \tau_d},
$$
where 
$$
\rho_{\tau}=\max_{\|n\|_{\infty}\leq (2d+\delta)\tau_1\cdots \tau_d}|V_{\tau}(n)-V(n)|;\ \ M_{\tau}= \max_{\|n\|_{\infty}\leq (2d+\delta)\tau_1\cdots \tau_d}|V(n)|.
$$
Then the equation $Hu=\lambda u$ with any $|\lambda|\leq \lambda_0$ does not have non-trivial $l^2(\Z^d)$ solutions.
\end{thm}

\section{Proof of Theorem \ref{genericae}}
For any $1>\delta>0$, let $\Gamma_n=(2d+\delta)\prod_{j=1}^{d} \tau_j^{(n)}$. Take $\gamma>1$. For $x, y\in \R^d$, let $x\star y=(x_1y_1,...,x_dy_d)$.
\begin{proof}
For any $k\in \N$, take $n_k$ such that 
\begin{align}\label{choosenk}
\frac{\Gamma_{n_k}}{\tau_i^{(n_k)}}\|\tau_i^{(n_k)}\alpha_i\|_{\T}<\frac{1}{k^2}\ \  \mathrm{for}\ \ i=1,...,d.
\end{align}
By Rokhlin-Halmos Lemma (see Theorem 1 on p.242 in \cite{CFS}), there exists a set $O_{n_k}$ such that $\{O_{n_k}+j\star \alpha\}_{\|j\|_{\infty}\leq k\Gamma_{n_k}}$ are disjoint and 
\begin{align*}
|O_{n_k}|>\frac{1-2^{-k-1}}{(2k\Gamma_{n_k}+1)^d}.
\end{align*}
We further partition $O_{n_k}$ into sets $S_{n_k, l}$, $1\leq l\leq s_{n_k}$ such that
\begin{align}\label{diamS}
\diam{(S_{n_k, l})}<\frac{1}{k}.
\end{align}
Choose compact set $K_{n_k, l}\subset S_{n_k, l}$ such that
\begin{align}\label{measureK}
\sum_{l=1}^{s_{n_k}}|K_{n_k, l}|>\frac{1-2^{-k}}{(2k\Gamma_{n_k}+1)^d}.
\end{align}
For $0\leq m_i\leq \tau_{i}^{(n_k)}-1$, $i=1,...,d$, we define 
\begin{align*}
U_{n_k, l, m}=\bigcup_{|j_i|\leq {k\Gamma_{n_k}}/{\tau^{(n_k)}_i}}K_{n_k, l}+m\star \alpha+ j\star\tau^{(n_k)}\star\alpha.
\end{align*}
Then by (\ref{choosenk}) and (\ref{diamS}), we have
\begin{align}\label{diamU}
\diam{(U_{n_k, l, m})}<\frac{2\sqrt{d}+1}{k}\ \ \mathrm{for}\ \mathrm{any}\ l,m\ \mathrm{above}.
\end{align}
Set
\begin{align*}
F_{n_k}=\{f\in C(\T^d): \ f\ \mathrm{is}\ \mathrm{constant}\ \mathrm{on}\ \mathrm{each}\ U_{n_k, l, m}\},
\end{align*}
and let $\mathcal{F}_{n_k}$ be the $k^{-\frac{2d+\gamma}{2d+\delta}\Gamma_{n_k}}$ neighborhood of $F_{n_k}$ in $C(\T^d)$. 
Note that by (\ref{diamU}) and the fact that continuous function on $\T^d$ is uniformly continuous, we have for each $t\in \N$,
\begin{align*}
\bigcup_{k\geq t}\mathcal{F}_{n_k}
\end{align*}
is open and dense subset of $C(\T^d)$. Thus
\begin{align*}
\mathcal{F}=\bigcap_{t\geq 1}\bigcup_{k\geq t}\mathcal{F}_{n_k}
\end{align*}
is a dense $G_{\delta}$ subset of $C(\T^d)$.

If $f\in \mathcal{F}$, then $f\in \mathcal{F}_{\tilde{n}_{k}}$ for a subsequence $\{\tilde{n}_k\}$ of $\{n_k\}$. 
For $k>4d-1+2\|f\|_{0}$, let 
\begin{align*}
T_{\tilde{n}_k}=\bigcup_{1\leq l\leq s_{\tilde{n}_k}, \|j\|_{\infty}\leq (k-1)\Gamma_{\tilde{n}_k}}(K_{\tilde{n}_k, l}+j\star\alpha).
\end{align*}
Then by (\ref{measureK}),
\begin{align}\label{measureT}
|T_{\tilde{n}_k}|\geq (\frac{(2k-2)\Gamma_{\tilde{n}_k}+1}{2k\Gamma_{\tilde{n}_k}+1})^d (1-2^{-k})\gtrsim 1-2^{-k}\ \ \mathrm{as}\ \ k\rightarrow\infty,
\end{align}
and for any $x\in T_{\tilde{n}_k}$, we have
\begin{align}
\max_{j\in \Z^d, |j_i|\leq \Gamma_{\tilde{n}_k}/{\tau_i^{(\tilde{n}_k)}}}|f(x+j\star\tau^{(\tilde{n}_k)}\star\alpha)-f(x)|<k^{-\frac{2d+\gamma}{2d+\delta}\Gamma_{\tilde{n}_k}}<(4d-1+2\|f\|_0)^{-\frac{2d+\gamma}{2d+\delta}\Gamma_{\tilde{n}_k}}.
\end{align}
By Borel-Cantelli lemma, a.e. $x\in \T^d$ belongs to infinitely many $T_{\tilde{n}_k}$, thus Theorem \ref{gordon} implies that for any $|\lambda|\leq 2d+\|f\|_{0}$, the equation $H_xu=\lambda u$ has no $l^2(\Z^d)$ non-trivial solution. Absence of point spectrum then follows from the fact the norm of $H_x$ is $\leq 2d+\|f\|_{0}$.     $\hfill{} \Box$
\end{proof}

\section{Proof of Theorem \ref{2d}}\label{2dproof}
Theorem \ref{2d} follows from a quick combination of Theorem \ref{genericae} and the following Lemma.
\begin{lemma}\label{iff}
When $d=2$, (\ref{GenericA}) holds if and only if $(\alpha_1, \alpha_2)$ are not both of bounded type.
\end{lemma}
\begin{proof}
\subsection*{The ``if" direction}
Let $\{\frac{p_n^{(i)}}{q_n^{(i)}}\}$ be continued fraction approximants of $\alpha_i$, $i=1,2$. By (\ref{qnqn+1}), it suffices to prove the following lemma.
\begin{lemma}
For any $\epsilon>0$, there exists $m,n\in \N$ such that 
\begin{equation}
\max{(\frac{q^{(2)}_{m}}{q^{(1)}_{n+1}}, \frac{q^{(1)}_{n}}{q^{(2)}_{m+1}})}<\epsilon.
\end{equation}
\end{lemma}
We will argue by contradiction. Assume that for some $\epsilon_0>0$, we have $\max{(\frac{q^{(2)}_{m}}{q^{(1)}_{n+1}}, \frac{q^{(1)}_{n}}{q^{(2)}_{m+1}})}\geq \epsilon_0$ for any $(m, n)\in \N^2$. Fix any $n\in \N$, choose $m$ such that
\begin{equation*}
\frac{q_n^{(1)}}{q_{m+1}^{(2)}}<\epsilon_0\leq \frac{q_n^{(1)}}{q_m^{(2)}}.
\end{equation*}
Then by our assumption, we have $\frac{q_m^{(2)}}{q_{n+1}^{(1)}}\geq \epsilon_0$, thus
\begin{equation*}
\epsilon_0 q_{n+1}^{(1)}\leq q_m^{(2)}\leq \frac{q_n^{(1)}}{\epsilon_0}.
\end{equation*}
This immediately implies $q_{n+1}^{(1)}\leq \frac{1}{\epsilon_0^2}q_{n}^{(1)}$ for any $n\in \N$, which means $\alpha_1$ is of bounded type.
Similarly, we could show $\alpha_2$ is also of bounded type, which is a contradiction.$\hfill{} \Box$

\subsection*{The ``only if" direction}
We will again argue by contradiction. Assume $(\alpha_1, \alpha_2)$ are both of bounded type and that there exists $\mathcal{Q}=\{\tau^{(n)}=(\tau_1^{(n)}, \tau_2^{(n)})\}$ such that
\begin{align}\label{C1}
\left\lbrace
\begin{matrix}
\tau_2^{(n)}\|\tau_1^{(n)}\alpha_1\|_{\T}\rightarrow 0,\\
\tau_1^{(n)}\|\tau_2^{(n)}\alpha_2\|_{\T}\rightarrow 0.
\end{matrix}
\right.
\end{align}
However, by (\ref{boundedtype}), we have for some $C>0$,
\begin{align*}
\left\lbrace
\begin{matrix}
\tau_2^{(n)}\|\tau_1^{(n)}\alpha_1\|_{\T}\geq \frac{C\tau_2^{(n)}}{\tau_1^{(n)}},\\
\tau_1^{(n)}\|\tau_2^{(n)}\alpha_2\|_{\T}\geq \frac{C\tau_1^{(n)}}{\tau_2^{(n)}},
\end{matrix}
\right.
\end{align*}
which obviously can not converge to $0$ at the same time, contradicting (\ref{C1}). $\hfill{} \Box$
\end{proof}

\section{Proof of Theorem \ref{measuregeneric}}
We need to divide into two different cases: {\it Case 1}. $\|f\|_{\infty}=\infty$ or {\it Case 2}. $\|f\|_{\infty}<\infty$. Here we provide a detailed proof of {\it Case 1} and discuss briefly about {\it Case 2} at the end of this section.

{\it Case 1}.
Let $\{e_i\}_{i=1}^d$ be the standard basis for $\R^d$. For a number $M>0$, let $E(M)$ be defined as in (\ref{defEM}). For $\tau=(\tau_1,...\tau_d)\in \N^d$, let $M_{\tau}=\inf\{M: E(M)\leq(\tau_1\cdots \tau_d)^{-d}\}$. Let $\kappa$ be defined as in (\ref{defkappa}).
The following lemma yields Theorem \ref{measuregeneric} directly.
\begin{lemma}\label{aephase}
Suppose there exists an infinite sequence $\mathcal{Q}$ with $\lim_{\mathcal{Q}\ni \tau\rightarrow\infty}\tau_i\rightarrow\infty$ for $i=1,..,d$.
Then if the frequencies satisfy 
\begin{align}\label{measureassume}
\|(\tau_1\alpha_1,..., \tau_d\alpha_d)\|_{\T^d} <\kappa(M_{\tau}^{-(2d+\gamma)\tau_1\cdots \tau_d}, {(\tau_1\cdots \tau_d)}^{-2})
\end{align}
for some $\gamma>0$ and any $\tau\in \mathcal{Q}$,
the operator $H_{x}$ has no point spectrum for a.e. $x\in \T^d$.
\end{lemma}

\begin{proof}
For any $\tau\in \mathcal{Q}$. Let $X_j^\tau=F(j\star \tau \star \alpha,\ (|j_1|+\cdots |j_d|)M_{\tau}^{-(2d+\gamma)\tau_1\cdots \tau_d})$ and $j^{(i)}=e_i\star j$. 
One could check directly by trigonometric inequality that
\begin{align*}
|f(x+j\star\tau\star\alpha)-f(x)|\leq \sum_{i=1}^d |f(x+(j^{(1)}+\cdots +j^{(i)})\star\tau\star\alpha)-f(x+(j^{(1)}+\cdots +j^{(i-1)})\star\tau\star\alpha)|,
\end{align*}
where we set $j^{(1)}+\cdots +j^{(i-1)}=0$ when $i=1$. Thus
\begin{align}\label{X1}
X_j^\tau\subset F(j^{(i)}\star\tau\star\alpha, |j_i| M_{\tau}^{-(2d+\gamma)\tau_1\cdots \tau_d})+(j^{(1)}+\cdots+j^{(i-1)})\star \tau\star\alpha.
\end{align}
Since $\|e_i\star\tau\star\alpha\|_{\T^d}\leq \|\tau\star\alpha\|_{\T^d}\leq \kappa(M_{\tau}^{-(2d+\gamma)\tau_1\cdots \tau_d}, {(\tau_1\cdots \tau_d)}^{-2})$, again by trigonometric inequality, we have
\begin{align}\label{X2}
|F(j^{(i)}\star\tau\star\alpha, |j_i|M_{\tau}^{-(2d+\gamma)\tau_1\cdots \tau_d})|\leq |j_i|{(\tau_1\cdots \tau_d)}^{-2}.
\end{align}
Hence, putting (\ref{X1}) and (\ref{X2}) together, we have
\begin{equation}\label{Xmeasure}
|X_j^{\tau}|\leq \frac{|j_1|+\cdots +|j_d|}{\tau_1^3\cdots \tau_d^3}.
\end{equation}
Let $Y_j^{\tau}=\{ x\in \T^d:|f(x+j\star \tau \star \alpha)|>M_{\tau} \} =E(M_{\tau})-j\star \tau \star \alpha$ , obviously,
\begin{align}\label{Ymeasure}
|Y_j^{\tau}|=|E(M_{\tau})|.
\end{align} 
For any $\frac{\gamma}{2}>\delta>0$, let $I_{\tau}=\{j\in \Z^d:|j_i| \leq (2d+\delta) \tau_1 \cdots \tau_d / \tau_i\ \mathrm{for}\ i=1,2, \cdots, d \}$. We denote $Z^{\tau}=(\cup_{I_{\tau}} X_j^{\tau})\cup(\cup_{I_{\tau}} Y_j^{\tau})$.
Combining (\ref{Xmeasure}), (\ref{Ymeasure}), we get
\begin{equation}\label{Zmeasure}
|Z^{\tau}| \leq \sum _{j\in I_{\tau}} \frac{|j_1|+\cdots+|j_d|}{(\tau_1\cdots \tau_d)^2}+\frac{|I_{\tau}|}{(\tau_1\cdots \tau_d)^d} \rightarrow 0\ \mathrm{as}\ \mathcal{Q}\ni\tau\rightarrow\infty.
\end{equation}
By Borel-Cantelli lemma, for a.e. $x \in \T^d$, there is an infinite sequence $\{\tau\}_{\tau\in \mathcal{P}_x}$ such that $x\notin Z^{\tau}$ for any $\tau\in \mathcal{P}_x$. Define a $\tau$-periodic potential by setting $f_{\tau}(x+n\star\alpha)=f(x+m\star\alpha)$, where $n_j\equiv m_j$ ($\mod \tau_j$) with $0\leq m_j\leq \tau_j-1$. Then since $x\notin Z^{\tau}$, we have
\begin{equation}\label{X3}
\max_{\|n\|_{\infty}\leq (2d+\delta)\tau_1\cdots \tau_d} |f_{\tau}(x+n\star \alpha)-f(x+n\star\alpha)|< M_{\tau}^{-(2d+\gamma)\tau_1\cdots \tau_d},
\end{equation} 
where $M_{\tau}\geq \max_{\|n\|_{\infty}\leq (2d+\delta)\tau_1\cdots \tau_d}|f(x+n\star\alpha)|$. 
Note that since $\|f\|_{\infty}=\infty$, we have $\lim_{\mathcal{P}_x\ni \tau\rightarrow \infty}M_{\tau}=\infty$. Together with (\ref{X3}) this implies that for any $\lambda_0>0$, for $\tau\in \mathcal{P}_x$ large, we have
\begin{equation}\label{X4}
\max_{\|n\|_{\infty}\leq (2d+\delta)\tau_1\cdots \tau_d} |f_{\tau}(x+n\star \alpha)-f(x+n\star\alpha)|< (2d-1+M_{\tau}+\lambda_0)^{-(2d+\frac{\gamma}{2})\tau_1\cdots \tau_d}.
\end{equation}
By Theorem \ref{gordon}, $H_x$ has no point spectrum. $\hfill{} \Box$
\end{proof}

{\it Case 2}. Note that when $\|f\|_{\infty}<\infty$, one could choose $M=\|f\|_{\infty}$ so that $E(M)=0$. Then one can prove Lemma \ref{aephase} with (\ref{measureassume}) replaced by
\begin{align}\label{measureassume2}
\|(\tau_1\alpha_1,..., \tau_d\alpha_d)\|_{\T^d} <\kappa((4d-1+2M)^{-(2d+\gamma)\tau_1\cdots \tau_d}, {(\tau_1\cdots \tau_d)}^{-2}).
\end{align}
$\hfill{} \Box$

\section{Arithmetic condition for $d>2$}\label{rem}
\begin{lemma}
For any $\epsilon>0$, let 
\begin{align}
A_{m_1,...,m_d}^{\epsilon}=\{\alpha\in \T^d: \max_{i=1,...,d}(\frac{m_1\cdots m_d}{m_i}\|m_i\alpha\|_{\T})<\epsilon\}.
\end{align}
Then $A^{\epsilon}=\limsup A_{m_1,...,m_d}^{\epsilon}$ has Lebesgue measure zero.
\end{lemma}
This Lemma clearly implies that when $d\geq 3$, $\alpha$'s such that (\ref{GenericA}) holds form a Lebesgue measure zero set.
\begin{proof}
Clearly, $|A_{m_1,...,m_d}^{\epsilon}|=\frac{(2\epsilon)^d}{(m_1\cdots m_d)^{d-1}}$. Thus
\begin{align}
\sum_{m_i\in \N}|A_{m_1,...,m_d}^{\epsilon}|=(2\epsilon)^d (\sum_{m\in \N}\frac{1}{m^{d-1}})^d<\infty\ \ \mathrm{for}\ \ d\geq 3.
\end{align}
By Borel-Cantelli lemma, $|A^{\epsilon}|=0$. $\hfill{} \Box$
\end{proof}

\section*{acknowledgement}
This research was partially supported by the NSF DMS-1401204. We are grateful to Svetlana Jitomirskaya for her guidance and useful discussions. 
F.Y was supported by CSC of China (No.201406330007) and the NSFC (grant no. 11571327).

\end{document}